\documentclass[peerreviewca,onecolumn,draftcls]{IEEEtran}
\usepackage{amsmath}
\usepackage{mathtools}
\usepackage{amssymb}
\usepackage{graphicx}
\newtheorem{thm}{Theorem}
\newtheorem{defn}{Definition}
\newtheorem{rem}{Remark}
\newenvironment{proof}{{\em Proof\/}:}{\\
\rightline{$\blacksquare$}}
\ifCLASSINFOpdf
\else
\fi

\begin{document}
%
\title{Discovering the Markov network structure}
%
%
%
\author{Edith~Kov\'{a}cs
        and~Tam\'{a}s~Sz\'{a}ntai  
\thanks{E. Kov\'{a}cs is with the Department of Methodology,
Budapest College of Management, Budapest, Vill\'{a}nyi \'{u}t 11--13, 1114 Hungary
e-mail: kovacs.edith@avf.hu } 
\thanks{T. Sz\'{a}ntai is with Budapest University of Technology and Economics. }} 
\maketitle

\begin{abstract}
In this paper a new proof is given for the supermodularity of information content.
Using the decomposability of the information content an algorithm is given for discovering the Markov network
graph structure endowed by the pairwise Markov
property of a given probability distribution. A discrete probability distribution is given for which
the equivalence of Hammersley-Clifford theorem is fulfilled although some of the possible vector realizations
are taken on with zero probability. Our algorithm for discovering the pairwise Markov network is illustrated on this example, too.
\end{abstract}


%
\IEEEpeerreviewmaketitle

\section{Introduction}
%
%
%
%
\IEEEPARstart{M}{arkov} networks together with Bayesian networks are probabilistic graphical
models, widely used for handling probability distributions, which endow
conditional independences.

Discovering the structure of the Markov network is a central task for
many fields like machine learning and its applications, computational
biology, reasoning under uncertainty, disease diagnosis, econometrics,
psychology and others \cite{Schlu12}. Markov networks are also known as Markov Random
Fields (when the underlying graph is a lattice), or as undirected graphical
models since the underlying graph is an undirected one. For good
introductions to this field see \cite{Pea88}, \cite{La96}, \cite{KolFri09}.

The goal of structure learning is to discover conditional independences in
the multivariate data. Typically, the structure learning is addressed through the
following algorithms. One is by evaluating each feature, and the highest
scoring feature is added to the model. The search may follow a top down
(general to specific) strategy \cite{Mc03}, \cite{De97}, \cite{Mal12}, or bottom up
\cite{DaDo10}. Other algorithms are going to
approximate the real probability distribution fitting a junction tree
probability distribution with a given tree width see \cite{Mal91},
\cite{SzaKo08}, \cite{KoSza10}, \cite{SzaKo13}. These
were done in greedy way. These studies were
all about finding approximations of discrete probability distributions.

In \cite{KoSa96} is given a computationally effective method for learning Markov
network structure from data. The method is based on the use of L1
regularization on the weights of log linear model, which has the effect of
biasing the model toward solutions where parameters are zero. This kind of
formulation leads to a convex optimization problem in continuous space,
where it can be solved using efficient gradient methods.

Also there are methods based on statistical independence tests \cite{Spetal00}, Mutual information \cite{CoTo91}, Parson's $\chi ^2$ and $G^2$
\cite{Ag02}. The case of Gaussian data is discussed among others in \cite{Spetal00} and \cite{Whi90}.
Other independence based algorithms are in
\cite{BrMa07}, \cite{BrMa09}. Many of the papers
conclude that unfortunately, the problem of learning Markov networks when
the number of random variables is high remains a hard challenge.

In the present paper we give a polynomial time algorithm in the number of random variables,
for discovering the exact Markov network of a multivariate random vector having the pairwise Markov property.
For doing this we use the concept of the information content.

The paper contains 5 parts. After the introductory part we give a short
overview of the relation between the graph structure underlying the Markov network, the
related hypergraphs and junction trees. Then we prove the supermodularity of
the information content. In the third part we give an algorithm for finding
the exact Markov network structure. In
the fourth part we give an example on which the algorithm can be visualized.
In the last part we summarize our results and discuss on future work.

\section{The supermodularity of the information content}
Let $V=\left\{ 1,\ldots ,n\right\} $ be a set of vertices and ${\mathcal C}$ a
set of subsets of $V$ called {\it set of hyperedges}. A \textit{hypergraph}
consists of a set $V$ of vertices and a set ${\mathcal C}$ of hyperedges. We
denote the hyperedges by $C_i$. If two vertices
are in the same hyperedge they are adjacent, which means, the hyperedge of a
hyperhraph is a complete graph on the set of vertices contained in it.

The \textit{acyclic}{\normalsize \ }\textit{hypergraph} is a special type of
hypergraph which fulfills the following requirements:

\begin{itemize}
\item  Neither of the edges of ${\mathcal C}$ is a subset of another edge.

\item  There exists a numbering of edges for which the \textit{running
intersection property} is fulfilled: $\forall j\geq 2\quad \ \exists \ i<j:\
C_i\supset C_j\cap \left( C_1\cup \ldots \cup C_{j-1}\right) $. (Other
formulation is that for all hyperedges $C_i$ and $C_j$ with $i<j-1$,
$C_i \cap C_j \subset C_s \ \mbox{for all} \ s, i<s<j$.)
\end{itemize}

Let $S_j=C_j\cap \left( C_1\cup \ldots \cup C_{j-1}\right) $, for $j>1$ and $%
S_1=\phi $. Let $R_j=C_j\backslash S_j$. We say that $S_j$\textit{separates}
$R_j$ from $\left( C_1\cup \ldots \cup C_{j-1}\right) \backslash S_j$, and
call $S_j$ separator.

Now we link these concepts to the terminology of junction trees.

The junction tree is a special tree structure which is assigned to the
connected acyclic hypergraphs \cite{LaSp84}, \cite{LaSp88}. The nodes of the tree correspond
to the hyperedges of the connected acyclic hypergraph and are called clusters, the edges of the tree
correspond to the separator sets and called separators. The set of all
clusters is denoted by $\mathcal{C}$, the set of all separators is denoted by %
$\mathcal{S}$. The triplet $(V,{\mathcal C},{\mathcal S})$ defines a junction tree.

The concept of \textit{junction tree probability distribution} is related to
the junction tree graph on the set of indices $V=\left\{ 1,\ldots ,n\right\} $
of a random vector $\mathbf{X=}\left( X_1,\ldots ,X_n\right) ^T$
with a given probability distribution ${\rm P}(\mathbf{X})$ taking on
values in the carthezian product space $\mathbf{\Lambda }=\Lambda _1\times
\ldots \times \Lambda _n$.

A junction tree probability distribution is defined as a product and division of
marginal probability distributions of ${\rm P}(\mathbf{X})$ as follows:

\begin{equation}\label{equ:equ0}
{\rm P}_{{\rm J}}\left( \mathbf{X}\right) =\frac{\prod\limits_{C\in \mathcal{C}}{\rm P}\left(
\mathbf{X}_C\right) }{\prod\limits_{S\in \mathcal{S}}\left[ {\rm P}\left( \mathbf{X}_S\right)
\right] ^{\nu _S-1}},
\end{equation}
where $\mathcal{C}$ is the set of clusters of the
junction tree, $\mathcal S$ is the set of separators, $\nu_S$ is the number of those clusters
which are intimately connected by $S$.

On the index set $V$ one can define a lot of junction trees.
To each junction tree one can assign a junction tree probability distribution
using Formula (\ref{equ:equ0}).

Let us consider a random vector $\mathbf{X}=\left( X_1,\ldots ,X_n\right) ^T$, with the
set of indices $V=\left\{ 1,\ldots ,n\right\}$. 
The {\it graph of a Markov network} consists of a set of
nodes V, and a set of edges $E=\left\{ \left( i,j\right) |i,j\in V\right\}$. 
We say the graph structure associated to the Markov network has

\begin{itemize}
\item  the \textit{Pairwise Markov} (PM) property if $\forall i,j\in V$, $i$
not connected to $j$ implies that $X_i$ and $X_j$ are conditionally
independent given all the other random variables;

\item  the \textit{Local Markov} (LM) property if $\forall i\in V,$ and
$ {\rm Ne}\left( i\right) $ the neighborhood of node $i$
in the graph (the nodes connected with $i$) then $X_i$ is conditionally independent from
all $X_j$, $j\notin {\rm Ne}\left( i\right) $, given $X_k,k\in Ne\left( i\right) $;

\item  the \textit{Global Markov} (GM) property states that if in the graph
$\forall A,B,C\subset V$ and $C$ separates $A$ and $B$ in terms of graph then $\mathbf{X}_A$ 
and $\mathbf{X}_B$ are conditionally independent given $\mathbf{X}_C$, which means in terms of
probabilities that
\[{\rm P}\left( \mathbf{X}_{A\cup B\cup C}\right) =\frac{{\rm P}\left( \mathbf{X}_{A\cup C}\right) {\rm P}\left( \mathbf{X}_{A\cup C}\right) }{{\rm P}\left( \mathbf{X}_C\right) };
\]
\end{itemize}

\begin{rem}\label{rem:1}
If a junction tree probability distribution ${\rm P}_{{\rm J}}(\mathbf{X})$ is associated to
${\rm P}(\mathbf{X})$ then all realizations which occur with positive probability in ${\rm P}(\mathbf{X})$ will occur with positive probability in ${\rm P}_{{\rm J}}(\mathbf{X})$.
\end{rem}

\begin{thm}(\cite{SzaKo08})
\label{thm:2.1}
The Kullback-Leibler divergence between the true ${\rm P}(\mathbf{X})$ and a junction tree probability distribution
${\rm P}_{{\rm J}}(\mathbf{X})$, determined by the set of clusters $\mathcal{C}$ and the set of separators $\mathcal{S}$ is:
\begin{equation}\label{equ:equ1}
\begin{array}{l}
KL\left( {\rm P}\left( \mathbf{X}\right) ,{\rm P}_{{\rm J}}\left( \mathbf{X}\right) \right)
= \sum\limits_{i=1}^{n} H\left( X_i\right) - H\left( \mathbf{X}\right) -\\
-\left( \sum\limits_{C\in \mathcal{C}}I\left(
\mathbf{X}_C\right) -\sum\limits_{S\in \mathcal{S}}\left( \nu _S-1\right) I\left(
\mathbf{X}_S\right) \right),
\end{array}
\end{equation}
where  $\nu_S$ is
the number of the clusters separated intimately by $S$, $I(\mathbf{X}_{C})=\sum\limits_{i\in \mathcal{C}} H\left(X_i\right) - H\left(
\mathbf{X}_{C}\right)$ represents the {\it information content} of the random vector $\mathbf{X}_{C}$
and similarly $I(\mathbf{X}_{S})=\sum\limits_{i\in \mathcal{S}} H\left(X_i\right) - H\left(
\mathbf{X}_{S}\right)$ represents the {\it information content} of the random vector $\mathbf{X}_{S}$.
\end{thm}

In Formula (\ref{equ:equ1}) $\sum\limits_{i=1}^dH\left(
X_i\right) -H\left( \mathbf{X}\right) =I\left( \mathbf{X}\right) $ is the {\it information content} of ${\rm P}(\mathbf{X})$ and it is independent from the structure of
the junction tree. It is easy to see that minimizing the Kullback-Leibler
divergence means maximizing $I_{{\rm J}}(\mathbf{X})=\sum\limits_{C\in \mathcal{C}}I\left(
\mathbf{X}_C\right) -\sum\limits_{S\in \mathcal{S}}\left( \nu _S-1\right) I\left(
\mathbf{X}_S\right) $. We call this sum as \textit{ the weight of the junction tree probability distribution}. As larger
this weight is, as better the approximation given by the junction
tree probability distribution fits to the probability distribution ${\rm P}(\mathbf{X})$. It is well known that $KL=0$ if and only if ${\rm P}(\mathbf{X})={\rm P}_{{\rm J}}\left( \mathbf{X}\right)$.

Mutual information was introduced in 1949 by Shannon and Weaver \cite{ShaWea49} as
a measure of dependence between two random variables. This concept was
generalized in two main directions. One generalization was introduced by Mc Gill \cite{McG'54}
and called \textit{interaction information} or \textit{mutual information}, and it is based on
the concept of the conditional entropy. The other generalization was introduced by Watanabe \cite{Wa60} and
it was called \textit{total correlation} and \textit{multiinformation} by
Studeny and Veinerova \cite{StuVej98}, {\it information content} by the authors \cite{SzaKo08}. Another important point of view is that
multiinformation is a special case of Csisz\'{a}r's \textit{I-divergence} \cite{Csi}.

The present paper is regarded to the multiinfiormation but we call it \textit{information content}, to be
consistent with our earlier papers. For a nice overview about the importance and
properties of multiinformation see paper \cite{StuVej98}.

Let $\mathbf{X}^{T}=(X_{1},\ldots,X_{n})$ be a given random vector and $V=\{1,\ldots,n\}$ the set of indices of the random variables.

We denote by $\mathbf{X}_A$ a random vector $\left( X_{i_1},\ldots
,X_{i_k}\right) ^T$ with $\left\{ i_1,\ldots ,i_k\right\} =A\subset V$

\begin{defn}\label{def:2.1}
The function $I:2^V\rightarrow R_{+}$ given as:
\[
I\left( \mathbf{X}_A\right) =\left\{
\begin{array}{c}\vspace{3mm}
\sum\limits_{i\in A}H\left( X_i\right) -H\left( \mathbf{X}_A\right) \quad
\mbox{if }A\in 2^V\mbox{and }\left| A\right| \geq 2 \\
0\qquad \qquad \qquad \qquad \mbox{if }A\in 2^V\mbox{and }\left| A\right| <2
\end{array}
\right.
\]
is called \textit{information content} of the random vector $\mathbf{X}_A$,
where $H\left( \cdot \right) $ means the entropy of a random variable
or a random vector.
\end{defn}

There are many machine learning/data mining applications that exploit the fact that entropy, (conditional) mutual information are submodular set functions \cite{Fuj78}, for selecting features/structure learning \cite{NaBi04}, \cite{NaBi13}.

In the following we are going to prove that the information content is
supermodular. We emphasize here that the supermodularity of the information content (multiinformation) was already proved by Studeny using the concept of Imsets \cite{Stu05}. We will give here a short proof based on our Theorem 1. 
We remind now the definition of supermodularity:

\begin{defn}\label{def:2.2}
Let be $V$ a set. A function $f:2^V\rightarrow R_{+}$ is
called {\it supermodular} if it satisfies the following condition:
\[
\forall A,B\in 2^V\quad f\left( A\cup B\right) +f\left( A\cap B\right) \geq
f\left( A\right) +f\left( B\right)
\]
\end{defn}

\begin{thm}\label{thm:2.2}
The information content $f\left( A\right)=I\left( \mathbf{X}_A\right)$ is supermodular
relative to a given probability distribution ${\rm P}(\mathbf{X})$.
\end{thm}
\begin{proof}
Let any $C_1,C_2\in 2^V$, there are three cases, see Figure \ref{fig:1}.

\begin{figure}[h]
\centering
  \includegraphics[bb=80 660 550 770,width=8cm]{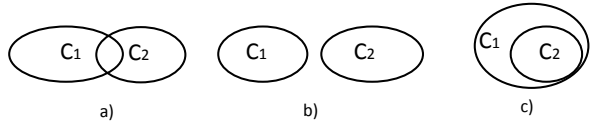}
\caption{The three possible relative positions of the sets $C_1$ and $C_2.$}
\label{fig:1}       
\end{figure}

In the case a) of Figure \ref{fig:1} the two clusters (hyperedges) form a junction tree,
therefore we can define a probability distribution as

\[
{\rm P}_{{\rm J}}\left( \mathbf{X}_{C_1\cup C_2}\right) =\frac{{\rm P}\left( \mathbf{X}%
_{C_1}\right) {\rm P}\left( \mathbf{X}_{C_2}\right) }{{\rm P}\left( \mathbf{X}%
_{C_1\cap C_2}\right) },
\]
where ${\rm P}\left( \mathbf{X}_{C_1}\right),\ {\rm P}\left( \mathbf{X}_{C_2}\right) \mbox{ and }
 {\rm P}\left( \mathbf{X}_{C_1\cap C_2}\right)$
are marginal probability distributions of ${\rm P}(\mathbf{X})$.

Using the expression of the Kullback-Leibler divergence given in Theorem \ref{thm:2.1} we have:
\begin{equation}\label{equ:equ2}
\begin{array}{l}\vspace{3mm}
KL\left( {\rm P}\left( \mathbf{X}_{C_1\cup C_2}\right) ,{\rm P}_{{\rm J}}\left( \mathbf{X}%
_{C_1\cup C_2}\right) \right)
= I\left( \mathbf{X}_{C_1\cup C_2}\right) - \\
-\left( I\left( \mathbf{X}_{C_1}\right) +I\left( \mathbf{X}_{C_2}\right)
-I\left( \mathbf{X}_{C_1\cap C_2}\right) \right) \geq 0
\end{array}
\end{equation}
and it equals to zero if and only if ${\rm P}\left( \mathbf{X}_{C_1\cup C_2}\right) $
has a junction tree structure assigned to $C_1$ and $C_2$.

Now by transforming (\ref{equ:equ2}) we obtain:
\[
I\left( \mathbf{X}_{C_1\cup C_2}\right) +I\left( \mathbf{X}_{C_1\cap
C_2}\right) \geq I\left( \mathbf{X}_{C_1}\right) +I\left( \mathbf{X}%
_{C_2}\right).
\]
In the case b) of Figure \ref{fig:1} the variables are contained in two different clusters.
We can define the following probability distribution:
\[
{\rm P}_{{\rm ind}}\left( \mathbf{X}_{C_1\cup C_2}\right) ={\rm P}\left( \mathbf{X}%
_{C_1}\right) {\rm P}\left( \mathbf{X}_{C_2}\right).
\]
By applying the abbreviation ${\rm P}(\mathbf{x}) = {\rm P}(\mathbf{X}=\mathbf{x})$ the Kullback-Leibler divergence will be:
\[
\begin{array}{l}\vspace{3mm}
KL\left( {\rm P}\left( \mathbf{X}_{C_1\cup C_2}\right) ,{\rm P}_{{\rm ind}}\left( \mathbf{X}%
_{C_1\cup C_2}\right) \right) = \\\vspace{3mm}
= \sum\limits_{\mathbf{x_{C_1\cup C_2}\in \Lambda }%
_{C_1}\times \mathbf{\Lambda }_{C_2}}{\rm P}\left( \mathbf{x}_{C_1\cup
C_2}\right) \log_2\frac{{\rm P}\left( \mathbf{x}_{C_1\cup C_2}\right)}
{{\rm P}\left( \mathbf{x}_{C_1}\right) {\rm P}\left( \mathbf{x}_{C_2}\right)}=\\\vspace{3mm}
= \sum\limits_{\mathbf{x_{C_1\cup C_2}\in\Lambda }_{C_1}\times \mathbf{\Lambda }_{C_2}}{\rm P}\left( \mathbf{x}_{C_1\cup
C_2}\right) \log_2 {\rm P}\left( \mathbf{x}_{C_1\cup C_2}\right) - \\\vspace{3mm}
- \sum\limits_{\mathbf{x_{C_{1}\cup C_{2}}} \in \Lambda_{C_{1}}\times \mathbf{\Lambda }%
_{C_{2}}}{\rm P}\left( \mathbf{x}_{C_1\cup C_2}\right) \log_2 \left[{\rm P}\left( \mathbf{x}%
_{C_1}\right) {\rm P}\left( \mathbf{x}_{C_2}\right) \right].
\end{array}
\]

The first sum equals to the negative entropy $-H\left( \mathbf{X}_{C_1\cup C_2}\right)$.

The second sum can be decomposed and then, using the fact that $C_1\cap C_2=\emptyset$, the two terms can be summed up according to all possible values of vector $\mathbf{x}_{c_2}$ resp. $\mathbf{x}_{c_1}$ as follows.
\[
\begin{array}{l}\vspace{3mm}
-\sum\limits_{\mathbf{x_{C_1\cup C_2}\in \Lambda }_{C_1}\times \mathbf{\Lambda }%
_{C_2}}{\rm P}\left( \mathbf{x}_{C_1\cup C_2}\right) \log_2 \left[{\rm P}\left( \mathbf{x}%
_{C_1}\right) {\rm P}\left( \mathbf{x}_{C_2}\right)\right] = \\\vspace{3mm}
= -\sum\limits_{\mathbf{x_{C_1\cup C_2}\in
\Lambda }_{C_1}\times \mathbf{\Lambda }_{C_2}}{\rm P}\left( \mathbf{x}_{C_1\cup
C_2}\right) \log_2 {\rm P}\left( \mathbf{x}_{C_1}\right) - \\\vspace{3mm}
- \sum\limits_{\mathbf{x_{C_1\cup C_2}\in
\Lambda }_{C_1}\times \mathbf{\Lambda }_{C_2}}{\rm P}\left( \mathbf{x}_{C_1\cup
C_2}\right) \log_2 {\rm P}\left( \mathbf{x}_{C_2}\right) = \\\vspace{3mm}
=-\sum\limits_{\mathbf{x_{C_1}\in \Lambda }_{C_1}}{\rm P}\left( \mathbf{x}_{C_1}\right)
\log_2 {\rm P}\left( \mathbf{x}_{C_1}\right) - \\\vspace{3mm}
- \sum\limits_{\mathbf{x_{C_2}\in \Lambda }%
_{C_2}}{\rm P}\left( \mathbf{x}_{C_2}\right) \log_2 {\rm P}\left( \mathbf{x}_{C_2}\right)= \\\vspace{3mm}
=H\left( \mathbf{X}_{C_1}\right) +H\left( \mathbf{X}_{C_2}\right).
\end{array}
\]
Returning to the expression of the Kullback-Leibler divergence we have:
\[
\begin{array}{l}\vspace{3mm}
KL\left( {\rm P}\left( \mathbf{X}_{C_1\cup C_2}\right) ,{\rm P}_{{\rm ind}}\left( \mathbf{X}%
_{C_1\cup C_2}\right) \right) = \\
= H\left( \mathbf{X}_{C_1}\right) +H\left(
\mathbf{X}_{C_2}\right) -H\left( \mathbf{X}_{C_1\cup C_2}\right).
\end{array}
\]
By adding and subtracting $\sum\limits_{i\in C_1}H\left( X_i\right) $ and $\sum\limits_{i\in C_2}H\left( X_i\right)$ we get:
\[
\begin{array}{l}\vspace{3mm}
KL\left( {\rm P}\left( \mathbf{X}_{C_1\cup C_2}\right) ,{\rm P}_{{\rm ind}}\left( \mathbf{X}%
_{C_1\cup C_2}\right) \right) = \\\vspace{3mm}
= H\left( \mathbf{X}_{C_1}\right)
-\sum\limits_{i\in C_1}H\left( X_i\right) +H\left( \mathbf{X}_{C_2}\right)
-\sum\limits_{i\in C_2}H\left( X_i\right) - \\\vspace{3mm}
- H\left( \mathbf{X}_{C_1\cup
C_2}\right) +\sum\limits_{i\in C_1}H\left( X_i\right) +\sum\limits_{i\in
C_2}H\left( X_i\right).
\end{array}
\]
Taking into account that in this case $C_1\cap C_2=\emptyset $:
\[
\sum\limits_{i\in C_1}H\left( X_i\right) +\sum\limits_{i\in C_2}H\left(
X_i\right) =\sum\limits_{i\in C_1\cup C_2}H\left( X_i\right)
\]
and so
\[
\begin{array}{l}\vspace{3mm}\vspace{3mm}
KL\left( {\rm P}\left( \mathbf{X}_{C_1\cup C_2}\right) ,{\rm P}_{{\rm ind}}\left( \mathbf{X}%
_{C_1\cup C_2}\right) \right) = \\\vspace{3mm}
= -\left[ \sum\limits_{i\in C_1}H\left( X_i\right) - H\left( \mathbf{X}_{C_1}\right)\right] - \\\vspace{3mm}
  -\left[ \sum\limits_{i\in C_2}H\left( X_i\right) - H\left( \mathbf{X}_{C_2}\right)\right] + \\\vspace{3mm}
+ \left[\sum\limits_{i\in C_1\cup C_2}H\left( X_i\right) -H\left( \mathbf{X}%
_{C_1\cup C_2}\right) \right]
\end{array}
\]
We express the Kullback-Leibler divergence using the information contents as follows:
\[
\begin{array}{l}\vspace{3mm}
KL\left( {\rm P}\left( \mathbf{X}_{C_1\cup C_2}\right) ,{\rm P}_{{\rm ind}}\left( \mathbf{X}%
_{C_1\cup C_2}\right) \right) = \\
= I\left( \mathbf{X}_{C_1\cup C_2}\right)
-I\left( \mathbf{X}_{C_1}\right) -I\left( \mathbf{X}_{C_2}\right) \geq 0
\end{array}
\]
Since $\left| C_1\cap C_2\right| = 0$, by Definition \ref{def:2.1} $I\left( \mathbf{X}_{C_1\cap C_2}\right) = 0$.
Taking into account this we obtain
\[
I\left( \mathbf{X}_{C_1\cup C_2}\right) +I\left( \mathbf{X}_{C_1\cap
C_2}\right) \geq I\left( \mathbf{X}_{C_1}\right) +I\left( \mathbf{X}%
_{C_2}\right).
\]

The case c) of Figure \ref{fig:1} is trivial because $C_1\cup C_2=C_1$ and $C_1\cap C_2=C_2$
so the inequality will be in fact an equality.
\end{proof}

\section{Discovering the exact structure of the Markov network}
Let $\mathbf{X=}\left( X_1,\ldots ,X_n\right) ^T$ be a random vector
taking values in the carthezian product space $\mathbf{\Lambda }=\Lambda
_1\times \ldots \times \Lambda _n$, and $V=\left\{ 1,\ldots ,n\right\} $ the
set of indices. Let us suppose its joint probability distribution ${\rm P}(\mathbf{X})$ is known. The
probability distribution is not supposed to take on necessarily its all possible realizations defined by $\Lambda _1\times \ldots \times \Lambda _n$ with positive probabilities. We give an algorithm which is polynomial in the
number of variables and discovers the Markov network endowed by the
pairwise Markov property.

Let us consider a complete graph on $V=\left\{ 1,\ldots ,n\right\} $.

If the edge $(i,j)$ is deleted from the complete graph defined on $V$ we
obtain two vertices which are not adjacent. In fact we obtain two clusters of $n-1$
elements such that $n-2$ of them are common. This represents a junction
tree. We denote the two marginal probability distributions assigned to the two clusters
$C_1=V \setminus \{i\}$ and $C_2=V \setminus \{j\}$ by ${\rm P}\left( \mathbf{X}_{C_1}\right)$ and
${\rm P}\left( \mathbf{X}_{C_2}\right) $.

\begin{thm}\label{thm:3.1}
In the Markov network of ${\rm P}(\mathbf{X})$, $X_i$ and $X_j$ are conditionally
independent with respect to the other random variables if and only if
\[
I\left( \mathbf{X}\right) =I\left( \mathbf{X}_{C_1}\right) +I\left(
\mathbf{X}_{C_2}\right) -I\left( \mathbf{X}_{C_1\cap C_2}\right)
\]
\end{thm}
\begin{proof}
We assign to the junction tree obtained by deleting the edge (i,j) a
junction tree probability distribution defined by:
\begin{equation}\label{equ:equ3}
{\rm P}_{{\rm J}}^{i\nsim j}\left( \mathbf{X}_{C_1\cup C_2}\right) =\frac{{\rm P}\left(
\mathbf{X}_{C_1}\right) {\rm P}\left( \mathbf{X}_{C_2}\right) }{{\rm P}\left(
\mathbf{X}_{C_1\cap C_2}\right) }
\end{equation}
where $i\nsim j$ in the upper index means that the edge $\left( i,j\right) $
is deleted.

Using Theorem \ref{thm:2.1} the Kullback-Leibler divergence can be written as:
\[
\begin{array}{l}\vspace{3mm}
KL\left( {\rm P}\left( \mathbf{X}\right) ,{\rm P}_{{\rm J}}^{i\nsim j}\left( \mathbf{X}%
_{C_1\cup C_2}\right) \right) = I(\mathbf{X})-I_{{\rm J}}(\mathbf{X})=\\
= I\left( \mathbf{X}\right) -\left( I\left(
\mathbf{X}_{C_1}\right) +I\left( \mathbf{X}_{C_2}\right) -I\left(
\mathbf{X}_{C_1\cap C_2}\right) \right) \geq 0.
\end{array}
\]

The equality holds if and only if the true probability distribution can be
written as in equation (\ref{equ:equ3}) what is equivalent to the conditional independence of $X_i$ and $X_j$
given all the other random variables and 
this is exactly the pairwise Markov property.
\end{proof}
\vspace{1mm}

\textit{Algorithm for discovering the pairwise Markov network}\vspace{3mm}

\textit{Input:} The probability distribution ${\rm P}\left( \mathbf{X}\right)$.

\textit{Output:} The pairwise Markov graph.

\text{1)} Calculate all of the $n-1$ dimensional marginal probability distributions of ${\rm P}\left(
\mathbf{X}\right)$ and their information contents.

\text{2)} Calculate all of the $n-2$ dimansional marginal probability distributions of ${\rm P}\left(
\mathbf{X}\right) $ and their information contents.

\text{3)} For each pair $(i,j),\ i\in V,\ j\in V,\ i\neq j$ let ${\rm P}(\mathbf{X}_{C_1}) $and ${\rm P}(\mathbf{X}_{C_2})$ the marginal probability distributions assigned to $C_{1}=V\setminus \left\{i\right\}$ and $C_{2}=V\setminus \left\{j\right\}$, test the following
equality:
\[
\begin{array}{rcl}\vspace{3mm}
KL&=&I(\mathbf{X})-I_{{\rm J}}(\mathbf{X})=\\
&=&I\left( \mathbf{X}\right) -\left[I\left( \mathbf{X}_{C_1}\right) +I\left(
\mathbf{X}_{C_2}\right) -I\left( \mathbf{X}_{C_1\cap C_2}\right)\right]
\end{array}
\]

If it equals 0, then $i$ and $j$ are not adjacent;
otherwise they are adjacent.
\vspace*{3mm}

In the following we remind some implications between different properties of the Markov network.
It is well known that ${\rm GM} \Rightarrow {\rm LM} \Rightarrow {\rm PM}$. The equivalence between them was proved in the Hammersley-Clifford theorem under the assumption that all vector realizations defined by the chartesian product
$\mathbf{x\in \Lambda }_1\times \ldots \times \mathbf{\Lambda }_n$ are taken on with positive probability (positivity condition).
Based on this, our algorithm is useful also for discovering the informations given by the ${\rm LM}$ and ${\rm GM}$ properties when positivity condition is fulfilled. However the positivity assumption is not necessary, therefore the following question arises. For which class of probability distributions holds the equivalence of the three Markov properties (${\rm GM} \Leftrightarrow {\rm LM} \Leftrightarrow {\rm PM}$)?

Let us regard a junction tree probability distribution. In this
case the probability distribution factorizes relative to a junction tree via
Formula (\ref{equ:equ0}). If we assign to the junction tree a graph structure such that two
vertices sharing the same cluster are connected, we obtain a triangulated
graph and call it {\it junction tree graph}.

\begin{defn}\label{def:3.1}
We say that the Markov network has the {\it Markov junction
tree property} (MJ) if the probability distribution factorizes relative to
the junction tree graph via Formula  (\ref{equ:equ0}).
\end{defn}
\begin{defn}\label{def:3.2}
We say that a junction tree probability distribution is
\textit{saturated} when there is no cluster such that two vertices belonging
to the same cluster are conditionally independent given all the other
vertices.
\end{defn}

The saturation condition is essential for ${\rm MJ}\Leftrightarrow {\rm PM}$,
but it is not necessary for ${\rm GM}\Leftrightarrow {\rm LM}\Leftrightarrow {\rm PM}$. 

In other words we say that a junction tree probability distribution is
\textit{saturated }when it factorizes relative to a triangulated graph
endowed by the pairwise Markov property.

\begin{rem}\label{rem:2}
Obviously the junction tree graph of junction tree probability
distribution is endowed by the global Markov property, so ${\rm MJ}\Rightarrow
{\rm GM}\Rightarrow {\rm LM}\Rightarrow {\rm PM}$.
\end{rem}
\begin{thm}\label{thm:3.2}
For a junction tree probability distribution
which is saturated the  pairwise Markov property implies Markov junction
tree property.
\end{thm}

In the next part we give an example where ${\rm MJ}\Leftrightarrow
{\rm GM}\Leftrightarrow {\rm LM}\Leftrightarrow {\rm PM}$ without fulfilling the positivity condition.

For illustrating this we recall Mousourris's sound
example \cite{Mou74}. The pairwise Markov network is given in Figure 2.

\begin{figure}[h]
\centering
  \includegraphics[bb=95 645 240 780,width=1.5cm]{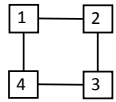}
\caption{The pairwise Markov network corresponding to Mousourris's example.}
\label{fig:2}       
\end{figure}

Since the probability distribution can be expressed as a junction tree probability distribution 
${\rm P}_{{\rm J}}^{i\nsim j}\left( \mathbf{X} \right)$ 
relative to $ C_1=\{1, 2, 4\};C_2=\{2, 3, 4\} $, respectively to $C_1=\{1, 3, 4\};C_2=\{1, 2, 3\}$ (these follow from
the equalities $KL=I\left( \mathbf{X}\right) - \left[I\left( \mathbf{X}%
_{C_1}\right) +I\left( \mathbf{X}_{C_2}\right) -I\left( \mathbf{X}%
_{C_1\cap C_2}\right) \right]=0$ for both cases, see Table I). 
The conditional independences $X_1\perp X_3|X_2,X_4$, and $X_2\perp X_4|X_1,X_3$ needed for
proving the global Markov property are verified, despite the junction tree probability distribution is not saturated.
In \cite{SadLau12} the reader can find other conditions on graphical models for which the equivalency of the three Markov properties hold.

\begin{table}[!ht]
\label{tab:1}
\caption{The information contents used for calculating the Kullback-Leibler divergences in the case of Mousourris's example}
\begin{center}
\begin{tabular}{|lc|c|}
\hline
\multicolumn{2}{|c|}{Information contents}       & K-L divergences \\\hline
$I(\mathbf{X})$                                & 1.000000 &  \\\hline
$I\left(\mathbf{X}_{V\setminus\{4\}}\right)$   & 0.500000 &  \\
$I\left(\mathbf{X}_{V\setminus\{3\}}\right)$   & 0.500000 &  \\
$I\left(\mathbf{X}_{V\setminus\{2\}}\right)$   & 0.500000 &  \\
$I\left(\mathbf{X}_{V\setminus\{1\}}\right)$   & 0.500000 &  \\\hline
$I\left(\mathbf{X}_{V\setminus\{3,4\}}\right)$ & 0.188722 & 0.188722   \\
$I\left(\mathbf{X}_{V\setminus\{2,4\}}\right)$ & 0.000000 & 0.000000   \\
$I\left(\mathbf{X}_{V\setminus\{2,3\}}\right)$ & 0.188722 & 0.188722   \\
$I\left(\mathbf{X}_{V\setminus\{1,4\}}\right)$ & 0.188722 & 0.188722   \\
$I\left(\mathbf{X}_{V\setminus\{1,3\}}\right)$ & 0.000000 & 0.000000   \\
$I\left(\mathbf{X}_{V\setminus\{1,2\}}\right)$ & 0.188722 & 0.188722   \\\hline
\end{tabular}
\end{center}
\end{table}

\section{A numerical example}
In this part we give an example of a random vector containing 8 random
variables. Each of them takes on two values. So the number of the possible realizations is $%
2^8=256$, but only 146 out of them are taken on with positive
probability. The entire probability distribution can be obtined from the authors by e-mail. We illustrate our algorithm on this example.

First using the information contents of all  $n-1=7$ variables and all $n-2=6
$ variables we apply our algorithm, and obtain the set of missing edges as result.
This can be followed in Table II.

\begin{table}[!ht]
\label{tab:2}
\caption{The information contents used for calculating the Kullback-Leibler divergences}
\begin{center}
\begin{tabular}{|lc|c|}
\hline
\multicolumn{2}{|c|}{Information contents}       & K-L divergences \\\hline
$I(\mathbf{X})$                                & 0.766387 &  \\\hline
$I\left(\mathbf{X}_{V\setminus\{8\}}\right)$   & 0.567413 &  \\
$I\left(\mathbf{X}_{V\setminus\{7\}}\right)$   & 0.559660 &  \\
$I\left(\mathbf{X}_{V\setminus\{6\}}\right)$   & 0.669253 &  \\
$I\left(\mathbf{X}_{V\setminus\{5\}}\right)$   & 0.618435 &  \\
$I\left(\mathbf{X}_{V\setminus\{4\}}\right)$   & 0.569887 &  \\
$I\left(\mathbf{X}_{V\setminus\{3\}}\right)$   & 0.679611 &  \\
$I\left(\mathbf{X}_{V\setminus\{2\}}\right)$   & 0.422819 &  \\
$I\left(\mathbf{X}_{V\setminus\{1\}}\right)$   & 0.473500 &  \\\hline
$I\left(\mathbf{X}_{V\setminus\{7,8\}}\right)$ & 0.421383 & 0.060698   \\
$I\left(\mathbf{X}_{V\setminus\{6,8\}}\right)$ & 0.470279 & 0.000000   \\
$I\left(\mathbf{X}_{V\setminus\{6,7\}}\right)$ & 0.462526 & 0.000000   \\
$I\left(\mathbf{X}_{V\setminus\{5,8\}}\right)$ & 0.419460 & 0.000000   \\
$I\left(\mathbf{X}_{V\setminus\{5,7\}}\right)$ & 0.456535 & 0.044828   \\
$I\left(\mathbf{X}_{V\setminus\{5,6\}}\right)$ & 0.573129 & 0.051828   \\
$I\left(\mathbf{X}_{V\setminus\{4,8\}}\right)$ & 0.412183 & 0.202216   \\
$I\left(\mathbf{X}_{V\setminus\{4,7\}}\right)$ & 0.419360 & 0.056200   \\
$I\left(\mathbf{X}_{V\setminus\{4,6\}}\right)$ & 0.523509 & 0.050756   \\
$I\left(\mathbf{X}_{V\setminus\{4,5\}}\right)$ & 0.487051 & 0.065117   \\
$I\left(\mathbf{X}_{V\setminus\{3,8\}}\right)$ & 0.523165 & 0.042528   \\
$I\left(\mathbf{X}_{V\setminus\{3,7\}}\right)$ & 0.515070 & 0.042187   \\
$I\left(\mathbf{X}_{V\setminus\{3,6\}}\right)$ & 0.582477 & 0.000000   \\
$I\left(\mathbf{X}_{V\setminus\{3,5\}}\right)$ & 0.531658 & 0.000000   \\
$I\left(\mathbf{X}_{V\setminus\{3,4\}}\right)$ & 0.483111 & 0.000000   \\
$I\left(\mathbf{X}_{V\setminus\{2,8\}}\right)$ & 0.289114 & 0.065270   \\
$I\left(\mathbf{X}_{V\setminus\{2,7\}}\right)$ & 0.231258 & 0.015167   \\
$I\left(\mathbf{X}_{V\setminus\{2,6\}}\right)$ & 0.325685 & 0.000000   \\
$I\left(\mathbf{X}_{V\setminus\{2,5\}}\right)$ & 0.274866 & 0.000000   \\
$I\left(\mathbf{X}_{V\setminus\{2,4\}}\right)$ & 0.226318 & 0.000000   \\
$I\left(\mathbf{X}_{V\setminus\{2,3\}}\right)$ & 0.336042 & 0.000000   \\
$I\left(\mathbf{X}_{V\setminus\{1,8\}}\right)$ & 0.300521 & 0.025995   \\
$I\left(\mathbf{X}_{V\setminus\{1,7\}}\right)$ & 0.266773 & 0.000000   \\
$I\left(\mathbf{X}_{V\setminus\{1,6\}}\right)$ & 0.376367 & 0.000000   \\
$I\left(\mathbf{X}_{V\setminus\{1,5\}}\right)$ & 0.325548 & 0.000000   \\
$I\left(\mathbf{X}_{V\setminus\{1,4\}}\right)$ & 0.277000 & 0.000000   \\
$I\left(\mathbf{X}_{V\setminus\{1,3\}}\right)$ & 0.386724 & 0.000000   \\
$I\left(\mathbf{X}_{V\setminus\{1,2\}}\right)$ & 0.408454 & 0.278523   \\\hline
\end{tabular}
\end{center}
\end{table}

In Figure 3 a) one can see the Markov graph structure endowed by the PM
property. Those $(i,j)$ edges are missing for which  
$KL\left( {\rm P}( \mathbf{X}\right) ,{\rm P}_{{\rm J}}^{i\nsim j}(\mathbf{X}_{C_1\cup C_2}))$  
(see Table II). 

\begin{figure}[h]
\centering
  \includegraphics[bb=80 480 510 785,width=8cm]{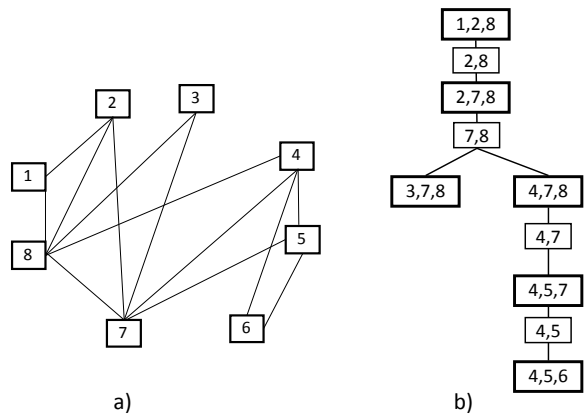}
\caption{The pairwise Markov network corresponding to Mousourris's example.}
\label{fig:3}       
\end{figure}

The junction tree assigned to ${\rm P}(\mathbf{X})$ is the one in Figure 3 b) since 
$KL\left( {\rm P}( \mathbf{X}\right) ,{\rm P}_{(\mathcal{C},\mathcal{S})}(\mathbf{X}))=0$
relative to the set of clusters  
$\mathcal{C}=\left\{ \left\{ 1,2,8 \right\}, \left\{ 2,7,8 \right\}, \left\{ 3,7,8 \right\},\right.$ $\left. \left\{ 4,7,8 \right\}, \left\{ 4,5,7 \right\}, \left\{ 4,5,6 \right\} \right\}$
and the set of separators  
$\mathcal{S}=\left\{  \left\{ 2,8 \right\}, \left\{ 7,8 \right\},\left\{ 7,8 \right\},\left\{ 4,7 \right\},\left\{ 4,5 \right\} \right\}$. 
It is easy to see that the junction tree in Figure 3 b) corresponds to the triangulated graph of the pairwise Markov network in Figure 3 a). This shows that in this example ${\rm GM}\Leftrightarrow {\rm PM}$  without satisfying the positivity condition.
\section{Conclusion}
In this paper we gave a new and short proof for the supermodularity of the information content based on one of our earlier theorems. Using this the problem of finding the best approximating junction tree for a given probability distribution can be reformulated as maximization a supermodular function on an independence set. This is a direction for future work. 
We gave a method for discovering the Markov network endowed by the pairwise Markov property of a given multivariate probability distribution. We believe that this method may have a great impact on researches in different fields, where discovering the dependence structure between the attributes is more complex and where experts have not preliminary insight to the problem. If the results have to be used on an empirical dataset then one should use results according to the asymptotic behaviour of empirical multiinformation as it is given in \cite{Stu87}.


\begin{thebibliography}{00}


\bibitem{Ag02}
A. Agresti, Categorical data analysis, 2nd edn., Wiley, New York, 2002.

\bibitem{BrMa07}
F. Bromberg, D. Margaritis, Efficient and robust independence-based Markov network structure discovery,
In: Proceedings of IJCAI, 2007.

\bibitem{BrMa09}
F. Bromberg, D. Margaritis, Efficient Markov network structure discovery using independence tests,
JAIR 35 (2009) 449-–485.

\bibitem{CoTo91}
T.M. Cover, J.A. Thomas, Elements of Information Theory,
Wiley Interscience, New York, 1991.

\bibitem{Csi}
I. Csisz\'{a}r, I–-divergence geometry of probability distributions and minimization
problems, Annals of Probability, 3 (1975), 146–-158.

\bibitem{DaDo10}
J. Davis, P. Domingos, Bottom-Up Learning of Markov Network Structure, In ICML, (2010)
271--278.

\bibitem{De97}
S. Della Pietra, V. J. Della Pietra, J. D. Lafferty, Inducing features of random fields, IEEE Trans PAMI,
19(4) (1997) 380-–393.


\bibitem{Fuj78} S. Fujishige, Polymatroidal dependence structure of a set of random variables, Information and Control, 39 (1978) 55-72.

\bibitem{Ki07}
S. Kirshner, Learning with tree-average densities and distributions,
Advances in Neural Information Processing Systems (NIPS), 2007.

\bibitem{KolFri09}
D. Koller, N. Friedman, Probabilistic Graphical Models: Principles and Techniques.
MIT Press, 2009.

\bibitem{KoSa96}
D. Koller, M. Sahami, Toward optimal feature selection, Morgan Kaufmann, Los Altos (1996) 284–-292.

\bibitem{KoSza10}
E. Kov\'{a}cs, T. Sz\'{a}ntai,
On the approximation of discrete multivariate probability distribution using the
new concept of $t$-cherry junction tree, Lecture Notes in Economics
and Mathematical Systems, 633, Proceedings of the
IFIP/IIASA/GAMM Workshop on Coping with Uncertainty, Robust Solutions, 2008,
IIASA, Laxenburg, 39--56.

\bibitem{LaSp84}
S.L. Lauritzen,  T. Speed, K. Vijayan,
Decomposable graphs and hypergraphs, J. Aust. Math. Soc. A 36 (1984) 12--29.

\bibitem{LaSp88}
S.L. Lauritzen, D.J. Spiegelhalter,
Local computations with probabilities on graphical structures and their application
to expert systems, Journal of the Royal Statistical Society, Ser. B 50 (1988) No. 2 157--224.

\bibitem{La96}
S. L. Lauritzen, Graphical models, Oxford University Press, Oxford, 1996.

\bibitem{Mal91}
F. M. Malvestuto, Approximating discrete probability distributions with decomposable
models, IEEE Trans. Systems, Man and Cybernetics, 21 (1991) 1287--1294.

\bibitem{Mal12}
F. M. Malvestuto, A backward selection procedure for approximating a discrete probability distribution by decomposable models, Kybernetika, 48 (2012) No. 5 825--844.

\bibitem{Mc03}
A. McCallum, Efficiently inducing features of conditional random fields, In: Proceedings of uncertainty
in artificial intelligence (UAI), 2003.

\bibitem{McG'54}
W. J. McGill, Multivariate information transmission, Psychometrika, 19 (1954) 97--116.

\bibitem{Ma05}
G. Mayor, J. Suner, J. Torrens,
Copula-like operations on finite settings,
IEEE Transactions on fuzzy systems 13 (2005) No. 4 468--477.

\bibitem{Mou74}
J. Moussouris, Gibbs and Markov random systems with constraints, Journal of Statistical Physics, 10(1) (1974) 11-–33.

\bibitem{NaBi04} M. Narasimhan and J. Bilmes, PAC learning bounded treewidth graphical models. Proc. 20th
annual conference on Uncertainty in Artificial Intelligence, 2004.

\bibitem{NaBi13}
M. Narasimhan J. Bilmes, A submodular-supermodular procedure with applications to
discriminative structure learning, arxiv/papers/1207/1207.1404.

\bibitem{Pea88}
J. Pearl, Probabilistic Reasoning in Intelligent Systems: Networks of Plausible Inference,
Morgan Kaufmann Publishers Inc., 1988.


\bibitem{SadLau12}
K. Sadeghi, S. Lauritzen, Markov properties for mixed graphs, Submitted to the Bernoulli,
arXiv: 1109.5909, 2012.


\bibitem{Schlu12}
F. Schl\"{u}ter, A survey on independence based Markov networks learning, Artificial Intelligence Revue,
DOI 10.1007/s10462-012-9346-y.

\bibitem{ShaWea49}
C. E. Shannon, W. Weaver, The mathematical theory of communication, University of Illinois Press, Urbana, Illinois, 1949. 

\bibitem{Spetal00}
P. Spirtes, C. Glymour, R. Scheines, Causation, prediction, and search, adaptive computation and machine
learning series, MIT Press, Cambridge, 2000.

\bibitem{Stu87} M. Studeny, Asymptotic behaviour of empirical multiinformation, Kybernetika, 23 (1987) No. 2.

\bibitem{StuVej98}
M. Studeny and J. Vejnarova, The multiinformation function as a tool for measuring stochastic dependence, In:  Learning in Graphical Models (M. I. Jordan ed.) Kluwer, Dordrecht, 1998, 261--298. 

\bibitem{Stu05} M. Studeny, Probabilistic Conditional Independence Structures, Springer-Verlag, London, 2005.

\bibitem{SzaKo08}
Sz\'{a}ntai, T. and E. Kov\'{a}cs,
Hypergraphs as a mean of discovering the dependence structure of a discrete multivariate probability
distribution, \textit{Proc. Conference APplied mathematical programming and
MODelling (APMOD), 2008}, Bratislava, 27-31 May 2008,
\textit{Annals of Operations Research}, 193 (2012) 71--90.

\bibitem{SzaKo13}
Sz\'{a}ntai, T. and E. Kov\'{a}cs,
Discovering a junction tree behind a Markov network by a greedy algorithm,
Optimization and Engineering, submitted, 2010.


\bibitem{WaJo08}
M. J. Wainwright, M. I. Jordan, Graphical models, exponential families, and variational inference, Found
Trends Mach Learn 1 (2008) 1-–305. doi:10.1561/2200000001

\bibitem{Wa60}
S. Watanabe, Information theoretical analysis of multivariate correlation, IBM Journal of Research and Development, 4 (1960) 66-–82.

\bibitem{Whi90}
J. Whittaker, Graphical Models in Applied Multivariate Statistics,
John Wiley \& Sons, Chichester, New York, Brisbane, Toronto, Singapore, 1990.

\end{thebibliography}
\end{document}